\documentclass{amsart}
\usepackage{amsfonts}
\usepackage{graphicx}
\usepackage{hyperref}
\newtheorem{theorem}{Theorem}[section]
\theoremstyle{plain}

\newtheorem{case}{Case}

\newtheorem{corollary}{Corollary}

\newtheorem{proposition}{Proposition}
\newtheorem{remark}{Remark}

\numberwithin{equation}{section}

\begin{document}
\title[Confidence bands ]{Asymptotic confidence bands for copulas based on the local linear kernel estimator}
\author{Diam Ba}
\address[Diam Ba]{LERSTAD, Universit\'e Gaston Berger, Saint-louis\\
S\'en\'egal}
\email[D. Ba]{diamba79@gmail.com}
\author{Cheikh Tidiane Seck}
\address[Cheikh Tidiane Seck]{LERSTAD, Universit\'e Alioune Diop, Bambey \\
S\'en\'egal}
\email[C. T. Seck]{cheikhtidiane.seck@uadb.edu.sn}
\author{Gane Samb Lo}
\address[Gane Samb Lo]{LERSTAD, Universit\'e Gaston Berger, Saint-louis\\
S\'en\'egal. LSTA, Universit\'e Pierre et Marie Curie, Paris, France}
\email[G. S. Lo]{pcsuvi@univi.org}

\date{\today}

\begin{abstract}
In this paper we establish asymptotic simultaneous confidence bands for copulas based on the local linear kernel estimator proposed by Chen and Huang  \cite{r1}. For this, we prove under smoothness conditions on the copula function, a uniform in bandwidth law of the iterated logarithm for the maximal deviation of this estimator from its expectation. We also show that the bias term converges uniformly to zero with a precise rate. The performance of these bands is illustrated in a simulation study. An application based on pseudo-panel data is also provided for modeling dependence.
\end{abstract}

\maketitle
  \section{Introduction}

Let us consider a random vector $\left(X,Y\right)$ with joint cumulative distribution function $H$ and marginal distribution functions $F$ and $G$. The Sklar's theorem (see \cite{r2}) says that there exists a bivariate distribution function $C$ on $[0,1]^2$  with uniform margins such that
 \begin{center}
$H(x,y)=C(F(x),G(y)).$
\end{center}
The function $C$ is called a copula associated with the random vector $(X,Y)$. If the marginal distribution functions $F$ and $G$ of $H$ are continuous, then the copula $C$ is unique and is defined as
$$C(u,v)=H(F^{\leftarrow}(u),G^{\leftarrow}(v)),$$ where $F^{\leftarrow}(u)=\inf\{x: F(x)\geq u\}$ and $G^{\leftarrow}(v)=\inf\{y: G(y)\geq v\}$, $\,\,u,v\in [0,1]$, are the generalized inverses of $F$ and $G$ respectively. \\

From these facts, estimating bivariate distribution function can be achieved in two steps : (i) estimating the margins $F$ and $G$; (ii) estimating the copula $C$.\\

 In this paper, we are dealing with nonparametric copula estimation. We consider a copula function $C$ with uniform margins $U$ and $V$ defined on $[0,1]$. Then, we can write 
$$ C(u,v)=\mathbb{P}(U\leq u, V\leq v),\;\; u,v\in[0,1].$$

The aim of this paper is to construct asymptotic optimal confidence bands, for the copula $C$, from the local linear kernel estimator proposed by Chen and Huang \cite{r1}. Our approach, based on  modern functional theory of empirical processes, allows the use of data-driven bandwidths for this estimator, and is largely inspired by the works of Mason \cite{r3} and Deheuvels and Mason \cite{r14}.\\
\indent There are two main methods for estimating copula functions : parametric and nonparametric methods. The Maximum likelihood estimation method (MLE) and the moment method are popular parametric approaches. It happens that one may use a nonparametric approach like the MLE-method and, at the same time, estimates margins by using parametric methods. Such an approach is called a  semi-parametric estimation method (see \cite{r4}). A popular nonparametric method is the kernel smoothing. Scaillet and Fermanian \cite{r5} presented the kernel smoothing method to estimate bivariate copulas for time series. 
Genest and Rivest \cite{r6} gave a nonparametric empirical distribution method to estimate bivariate Archimedean copulas.\\
 \indent A pure nonparametric estimation of copulas treats both the copula and the margins in a parameter-free way and thus offers the greatest generality. Nonparametric estimation of copulas goes back to Deheuvels \cite{r7} who proposed  an estimator based on a multivariate empirical distribution function and on its empirical marginals. Weak convergence studies of this estimator can be found in Fermanian et al. \cite{r8}. Gijbels and Mielniczuk \cite{r9} proposed  a kernel estimator for a bivariate copula density. Another approach of kernel estimation is to directly estimate a copula function  as explored in \cite{r5}. Chen and Huang \cite{r1} proposed a new bivariate kernel copula estimator by using local linear kernels and a simple mathematical correction that removes the boundary bias. They also derived the bias and the variance of their estimator, which reveal that the kernel smoothing produces a second order reduction in both the variance and mean square error as compared with the unsmoothed empirical estimator of Deheuvels \cite{r7}.\\
\indent Omelka, Gijbels and Veraverbeke \cite{r10} proposed improved \textit{shrinked} versions of the estimators of Gijbels and Mielniczuk \cite{r9} and Chen and Huang \cite{r1}. They have done this shrinkage by including a weight function that removes the corner bias problem. They also established weak convergence for all newly-proposed estimators.\\

In parallel, powerful technologies have been developed for density  and distribution function kernel estimation. We refer to Mason \cite{r3}, Dony \cite{r11}, Dony and Einmahl \cite{r12}, Einmahl and Mason \cite{r13}, Deheuvels and Mason \cite{r14}. 
In this paper we'll apply these recent methods to kernel-type estimators of copulas. The existence of kernel-type function estimators should lead to nonparametric estimation by confidence bands, as shown in \cite{r14}, where general asymptotic simultaneous confidence bands  are established for the density and the regression function curves. 
Furthermore, to our knowledge, there are not yet such type of results in nonparametric estimation of copulas. 
This motivated us to extend such technologies to kernel estimation of copulas by providing asymptotic simultaneous optimal confidence bands.  


Let $(X_1,Y_1),...,(X_n,Y_n)$ be an independent and identically distributed sample of the bivariate random vector $(X,Y)$, with continuous marginal cumulative distribution function $F$ and $G$. To construct their estimator, Chen and Huang proceed in two steps. In the first step, they estimate margins by
\begin{center} 
$\displaystyle \hat{F}_n(x)=\frac{1}{n}\sum_{i=1}^n K\left(\frac{x-X_i}{b_{n1}}\right),\ \ \ \ \ \hat{G}_n(y)=\frac{1}{n}\sum_{i=1}^n K\left(\frac{y-Y_i}{b_{n2}}\right),$
\end{center}  where $b_{n1}$ and $b_{n2}$ are some bandwidths and $K$ is the integral of a symmetric bounded kernel function $k$ supported on $[-1,1]$. In the second step, the pseudo-observations $\hat{U}_i = \hat{F}_n(X_i)$ and $\hat{V}_i =\hat{G}_n(Y_i)$ are used to estimate the joint distribution function of the unobserved $F(X_i)$ and $G(Y_i)$, which gives the estimate of the unknown copula $C$. To prevent boundary bias, Chen and Huang suggested using a local linear version of the kernel $k$ given by
$$
 k_{u,h}(t)=\frac{k(t)\{a_2(u,h)-a_1(u,h)t\}}{a_0(u,h)a_2(u,h)-a_1^2(u,h)}\mathbb{I}\left\{\frac{u-1}{h}<t<\frac{u}{h}\right\},$$
where $ a_j(u,h)=\int_{(u-1)/h}^{u/h} t^j k(t)dt$ for $j=0,1,2$ ; $u\in[0,1]$ and $0<h<1$ is a bandwidth. Finally, the local linear kernel estimator of the copula $C$ is defined as 
\begin{equation}\label{e1}
\hat{C}_{n,h}^{(LL)}(u,v)=\frac{1}{n}\sum_{i=1}^{n}K_{u,h}\left(\frac{u-\hat{U_i}}{h} \right)K_{v,h}\left(\frac{v-\hat{V_i}}{h} \right),
\end{equation}
where  $K_{u,h}(x)=\int_{-\infty}^{x}k_{u,h}(t)dt.$ The subscript $h$ in $\hat{C}_{n,h}^{(LL)}(u,v)$ is a variable bandwidth which may depend either on the sample data or the location $(u,v)$.\\

Our best achievement is the construction of asymptotic confidence bands from a uniform in bandwidth law of the iterated logarithm (LIL) for the maximal deviation of the  local linear estimator \eqref{e1},  and  the uniform convergence of the bias to zero with the same speed of convergence.


The paper is organized as follows. In Section $2$, we expose our  main results  in Theorems \ref{t1}, \ref{t2} and \ref{t3}. Simulation studies and  applications to real data sets are also made in this section  to illustrate  these results. In Section $3$, we report the proofs of  our assertions. The paper is ended by appendices in which we postpone some technical results and numerical computations.

\section{Main results and applications} 

\subsection{Results}

Here, we state our theoretical results. Theorem \ref{t1} gives a uniform in bandwidth LIL for the maximal deviation of the estimator \eqref{e1}. Theorem \ref{t2}  handles the  bias, while Theorem \ref{t3} provides  asymptotic optimal simultaneous  confidence bands for the copula function $C(u,v)$. 
\begin{theorem}\label{t1}
Suppose that the copula function $C(u,v)$ has bounded first order partial derivatives on $[0, 1]^2$. Then for any sequence of positive constants $(b_n)_{n\geq 1}$ satisfying $0<b_n<1, b_n\rightarrow 0$, $b_n\geq (\log n)^{-1}$, and for some $c>0$, we have almost surely
\begin{equation}\label{e2}
\limsup_{n\rightarrow\infty}\left\{R_n\sup_{\frac{c\log n}{n}\le h\le b_n}\sup_{(u,v)\in[0,1]^2}\left|\hat{C}_{n,h}^{(LL)}(u,v)-\mathbb{E}\hat{C}_{n,h}^{(LL)}(u,v)\right|\right\}=A(c),
\end{equation}
where $A(c)$ is a positive constant such that $0<A(c)\leq 3$ and $R_n =\left(\frac{n}{2\log\log n}\right)^{1/2}$.
\end{theorem} 
\begin{remark}
 Theorem \ref{t1} represents a uniform in bandwidth law of the iterated logarithm for the maximal deviation of the estimator (\ref{e1}). As in \cite{r14} we may use it, in its probability version, to construct simultaneous asymptotic confidence bands from the estimator \eqref{e1}. In this purpose, we must ensure before hand that the bias term $B_{n,h}(u,v)=\mathbb{E}\hat{C}_{n,h}^{(LL)}(u,v) - C(u,v)$ converges uniformly to $0$,  with the same rate  $R_n$, as $n\rightarrow\infty$. But this requires that the  copula function $C(u,v)$ admits bounded second-order partial derivatives on the unit square $[0,1]^2$.
\end{remark}

\begin{theorem} \label{t2}
Suppose that the copula function  $C(u,v)$ admits bounded second-order partial  derivatives on $[0, 1]^2$. Then for any sequence of positive constants $(b_n)_{n\geq 1}$ satisfying $0<b_n<1$, $\sqrt{n}b_n^2/\sqrt{\log\log n}=o(1),$ and for some $c>0$, we have almost surely,
\begin{equation} \label{biais}
R_n\sup_{\frac{c\log n}{n}\le h\le b_n}\sup_{(u,v)\in[0,1]^2}\vert \mathbb{E}\hat{C}_{n,h}^{(LL)}(u,v) - C(u,v)\vert \rightarrow 0,\,\, n\rightarrow\infty.
\end{equation}
\end{theorem}

\noindent \textbf{Useful comment.}
Because a number of copula families do not possess bounded second-order partial derivatives, the application of these results is limited by a corner bias problem. To overcome this difficulty and apply these results to a wide family of copulas, we adopt the shrinkage method of Omelka \textit{et al.} \cite{r10}, by taking a local data-driven bandwidth $\hat{h}_n(u, v)$ satisfying  the  following condition :  
$$
(H_1)\ \ \forall \epsilon >0,\ \mathbb{P}\left(\sup_{(u,v)\in[0,1]^2}\left|\frac{\hat{h}_n(u, v)}{h_n}-b(u, v)\right|>\epsilon\right)\longrightarrow 0,\ n\rightarrow \infty,
$$
where $h_n$ is a sequence of positive constants converging to $0$, and $b(u, v)$ is a real-valued function defined by
\begin{equation}\label{poi}
b(u,v)=\max\left\{\min\{u^\alpha,(1-u)^\alpha\},\min\{v^\alpha,(1-v)^\alpha\}\right\},\ \alpha>0.
\end{equation}
For such a bandwidth $\hat{h}_n(u, v)$, the local linear kernel estimator can be rewritten as
\begin{equation}\label{e3}
\hat{C}_{n,\hat{h}_n(u,v)}^{(LL)}(u,v)=\frac{1}{n}\sum_{i=1}^{n}K_{u,\hat{h}_n(u,1)}\left(\frac{u-\hat{U_i}}{\hat{h}_n(u,1)} \right)K_{v,\hat{h}_n(1,v)}\left(\frac{v-\hat{V_i}}{\hat{h}_n(1,v)} \right).
\end{equation}
By condition $(H_1)$,  \eqref{e3} is equivalent for $n$ large enough  to
\begin{equation}\label{e4}
\hat{C}_{n,{h}_n}^{(LL)}(u,v)=\frac{1}{n}\sum_{i=1}^{n}K_{u,{h}_n}\left(\frac{u-\hat{U_i}}{{h}_n b(u,1)} \right)K_{v,{h}_n}\left(\frac{v-\hat{V_i}}{{h}_n b(1,v)} \right).
\end{equation}
This latter estimator \eqref{e4} is exactly the improved "shrinked" version proposed by Omelka et al. \cite{r10}. It enables us to keep the bias bounded on the borders of the unit square and then  to remove the problem of possible unboundedness of the second order partial derivatives of the copula function $C$. To set up asymptotic optimal simultaneous confidence bands for the copula $C$, we  need the following additional condition :
$$
(H_2)\ \  \ \mathbb{P}\left(\frac{c\log n}{n}\le \hat{h}_n(u,v)\le b_n, \;\forall\; 0\le u,v\le 1,\right)\longrightarrow 1,\ n\rightarrow \infty.
$$

If  conditions $(H_1)$ and $(H_2)$ hold,  
then we can infer from Theorem \ref{t1} that
$$
 R_n\sup_{0\leq u,v\leq 1}\left|\hat{C}_{n,\hat{h}_n(u,v)}^{(LL)}(u,v)-\mathbb{E}\hat{C}_{n,\hat{h}_n(u,v)}^{(LL)}(u,v)\right|\stackrel{\mathbb{P}}{\longrightarrow}A(c),\ \ \ n\rightarrow \infty.
$$
This is still equivalent to
\begin{equation} \label{e5}
\sup_{0\leq u,v\leq 1}\frac{R_n}{A(c)}\left|\hat{C}_{n,\hat{h}_n(u,v)}^{(LL)}(u,v)-\mathbb{E}\hat{C}_{n,\hat{h}_n(u,v)}^{(LL)}(u,v)\right|\stackrel{\mathbb{P}}{\longrightarrow}1,\ \ \ n\rightarrow \infty.
\end{equation}

To make use of \eqref{e5} for forming confidence bands, we must ensure that the bandwidth $\hat{h}_n(u,v)$ is chosen in such a way that the bias of  the estimator (\ref{e3}) may be neglected , in the sense that
\begin{equation}\label{e6}
\sup_{0\leq u,v\leq 1}\frac{R_n}{A(c)}\left|\mathbb{E}\hat{C}_{n,\hat{h}_n(u,v)}^{(LL)}(u,v)-C(u,v)\right|\stackrel{\mathbb{P}}{\longrightarrow}0,\ \ \ n\rightarrow \infty.
\end{equation}
This would be the case if condition $(H2)$ holds and $\sqrt{n}b_n^2/\sqrt{\log\log n}=o(1)$. 

\begin{theorem}\label{t3}
Suppose that the assumptions of Theorem \ref{t1} and Theorem \ref{t2} hold. 
Then for any local data-driven bandwidth $\hat{h}_n(u,v)$ satisfying $(H_1)$ and $(H_2)$, and any  $\epsilon>0$, one has, as $n\rightarrow\infty$,
\begin{equation}\label{leq1}
\mathbb{P}\bigg(C(u,v)\in \left[\hat{C}_{n,\hat{h}_n(u,v)}^{(LL)}(u,v)-E_n(\epsilon),\hat{C}_{n,\hat{h}_n(u,v)}^{(LL)}(u,v)+E_n(\epsilon)\right], \forall\:0\leq u,v\leq 1 \bigg)\longrightarrow 1
\end{equation}
and,
\begin{equation}\label{leq2}
\mathbb{P}\bigg(C(u,v)\in \left[\hat{C}_{n,\hat{h}_n(u,v)}^{(LL)}(u,v)-\Delta_n(\epsilon),\hat{C}_{n,\hat{h}_n(u,v)}^{(LL)}(u,v)+\Delta_n(\epsilon)\right], \forall\:0\leq u,v\leq 1 \bigg)\longrightarrow 0,
\end{equation}
where $E_n(\epsilon)=(1+\epsilon)\frac{A(c)}{R_n}$, $\Delta_n(\epsilon)=(1-\epsilon)\frac{A(c)}{R_n}$, and $R_n =\left(\frac{n}{2\log\log n}\right)^{1/2}$.
\end{theorem}

\begin{remark}
Whenever (\ref{leq1}) and (\ref{leq2}) hold jointly for each $\epsilon > 0$, we will say that the intervals 
\begin{equation}\label{interv}
\left[\alpha_n(u,v),\beta_n(u,v)\right]= \left[\hat{C}_{n,\hat{h}_n(u,v)}^{(LL)}(u,v)-\frac{A(c)}{R_n}\:,\:\hat{C}_{n,\hat{h}_n(u,v)}^{(LL)}(u,v)+\frac{A(c)}{R_n}\right],
\end{equation}
provide asymptotic simultaneous optimal confidence bands (at an asymptotic confidence level of 100\% ) for the copula  function $C(u,v),\ 0\leq u,v\leq 1.$ So, with a probability near to 100\%, we can write for all $(u,v)\in [0,1]^2,$
\begin{equation}
 C(u,v)\in\left[\alpha_n(u,v),\beta_n(u,v)\right].
\end{equation}

\end{remark}

\subsection{Simulation results and data-driven applications}

\subsubsection{Simulation results}

We make some simulation studies to evaluate the performance of our asymptotic confidence bands. To this end, we compute  the confidence bands given in (\ref{interv}) for some classical parametric copulas, and check for whether the true copula is lying in these bands. For simplicity, we consider for example  two families of copulas : Frank and Clayton, defined respectively as follows :   
\begin{equation}\label{frk}
C_{\theta, F}(u,v)=-\frac{1}{\theta}\log\left\{1+\frac{\left(e^{-\theta u}-1\right)\left(e^{-\theta v}-1\right)}{e^{-\theta}-1}\right\},\quad \theta \in \mathbb{R}=(-\infty, \infty),
\end{equation}
and
\begin{equation}\label{clay}
C_{\theta,C}(u,v)=\left(u^{-\theta}+v^{-\theta}-1\right)^{-1/\theta},\quad \theta\in (0,\infty).
\end{equation}

We  fix values for the parameter $\theta$, and generate $n$ pairs of data : $(u_i,v_i),i=1,\cdots,n$, respectively from the two copulas by  using the conditional sampling method. The steps for drawing from a  bi-variate copula $C$ are :
\begin{itemize}
	\item step 1 : Generate two values $u$ and $v$ from $\mathcal{U}(0, 1)$,
	\item step 2 : Set $u_i=u$,
	\item step 3 : Compute $v_i=C_2^{-1}(v/u_i)$, where $C_2(v_i/u_i)= \frac{\partial}{\partial u_i} C(u_i, v_i)$.
\end{itemize}

Then  $u_i$ and  $v_i$ are random observations drawn from the copula $C$.

To compute the estimator $\hat{C}^{(LL)}_{n,\hat{h}_n(u,v)}(u,v)$, we take $h_n=1/\log n$ and $b(u,v)$ given by  formula \eqref{poi}, with $\alpha = 0.5$, so that conditions ($H_1)$ and $(H_2)$ are fulfilled. That is the case for $b_n=\left\{(\log\log n)/n\right\}^{1/4}.$ The function $K_{w,h}$  is obtained by integrating the
the local linear kernel function $k_{w,h_n}$ defined in the introduction, where $k$ is the Epanechnikov kernel density defined as  $k(x)=0.75(1-x^2)\mathbb{I}_{\left\{\left|x\right|\leq 1\right\}}$. 
Finally for any $(u,v)\in [0,1]^2$, we compute the confidence interval (\ref{interv})  by taking $A(c)=3$. \\
\indent In Figure 1, we represent the confidence bands and the Frank copula, while Figure 2 represents the confidence bands and the Clayton copula. One can see that the true curves of the two parametric copulas  are well contained in the bands. \\
\begin{figure}[ht]
\begin{center}
\includegraphics[width=8cm]{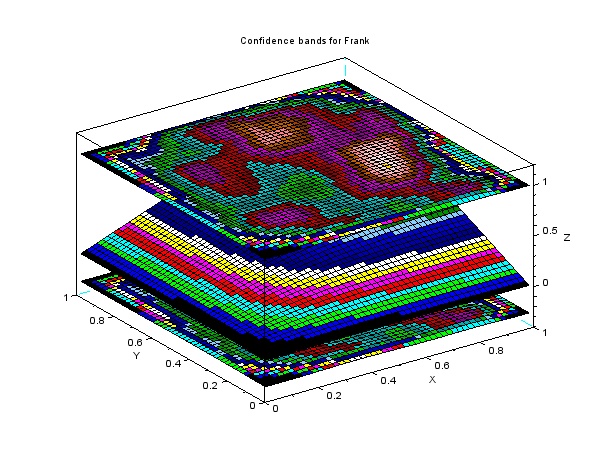}
\caption{Confidence bands for the Frank copula in 3D, with $\theta=5$ }
\label{fig11}
\end{center}
\end{figure}

\begin{figure}[ht]
\begin{center}
\includegraphics[width=8cm]{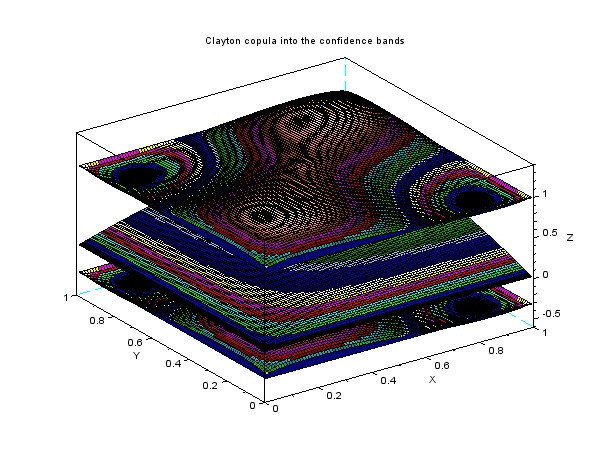}
\caption{Confidence bands for the Clayton copula in 3D, with $\theta=2$ }
\label{fig10}
\end{center}
\end{figure}

As we cannot visualize all the information in the above figures,  we provide in Appendix  some numerical computations to best appreciate the performance of our bands. To this end, we  generate 10 couples $(u,v)$ of random numbers uniformly distributed in $(0,1)$ and compute, for each of them and for each of the considered copulas, the lower bound $\alpha_n(u,v)$, the upper bound $\beta_n(u,v)$  and the true value of the copula $C(u,v)$ for different values of $\theta$. These computations are given in  Appendix B (see, Tables 3 and 4). We can see there, that all the values of $C(u,v)$ are in their respective confidence intervals.
\subsubsection{Data-driven applications}
In this subsection, we apply our theoretical results to select graphically, among various copula families, the one that best fits sample data. Towards this end, we shall represent in a same 2-dimensional graphic the confidence bands established in Theorem \ref{t3} and the curves corresponding to the different copulas considered. To illustrate this,  we use data expenses of senegalese households, available in databases managed by the National Agency of Statistics and Demography (ANSD) of the Republic of Senegal (www.ansd.sn). The data were obtained from  two sample surveys : ESAM2 (Senegalese Survey of Households, 2nd edition, 2001-2002) and  ESPS (Monitoring Survey of Poverty in Senegal, 2005-2006).  For simplicity, we shall deal with the pseudo-panel data utilized by \cite{r21}, which consist of two series of observations  of size $n=116$.
Instead of smoothing these observations denoted by $(X_i; Y_i),i=1,\cdots,n$, we deal with pseudo-observations 
$$\left(\hat{U}_i,\hat{V}_i\right)=\left(\frac{n}{n+1}F_n(X_i),\frac{n}{n+1}G_n(Y_i)\right),$$
to define the kernel estimator for the true copula. Here, $F_n$ and $G_n$ are empirical cumulative distribution functions associated respectively with the samples $X_1,\cdots,X_n$ and $Y_1,\cdots,Y_n$. \\
\indent This application is limited to Archimedean copulas. We will consider for example three parametric families of copulas : Frank, Gumbel and Clayton. Our aim is to find graphically, using our confidence bands, the family that best fits these pseudo-panel data. 
The unknown parameter $\theta$, for each family, is estimated by inversion of Kendall's tau.  For this, we first calculate the empirical Kendall's tau (we find $\hat{\tau} = 0.408 $), and then we deduce from it the values of the parameter $\theta$ for each family. 
\begin{table}[htbp]
\begin{center}
\begin{tabular}{|c|c|c|} 
\hline Copula & $\tau(\theta)$ &  $\theta$  \\
\hline Clayton & $\theta /(\theta+2)$ & 1.38  \\
\hline Gumbel & $(\theta-1)/\theta$ & 1.69  \\
\hline Frank & $1-\frac{4}{\theta}\left\{1-D_1(\theta)\right\} $ & -0.57  \\
\hline 
\end{tabular}
\end{center}
\caption{Expression of Kendall's tau and estimated values for $\theta$}
\label{tab1}
\end{table}
$D_1$ is the Debye function of order 1 defined as : $D_1(x)=\frac{1}{x}\int_0^x \frac{t}{e^t - 1 } dt$.\\
\begin{figure}[ht]
\begin{center}
\includegraphics[width=10cm]{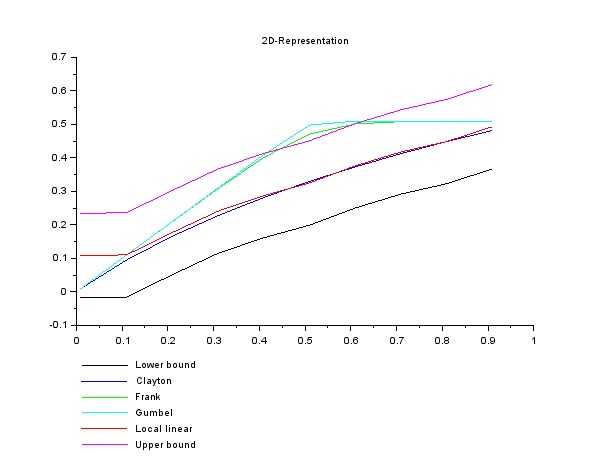}
\caption{Simultaneous representation of the three copulas into the confidence bands.}
\label{fig3}
\end{center}
\end{figure}
Figure 3 shows that the Clayton family seems to be more adequate to fit our pseudo-panel data. That is,  dependence fitting of these Senegalese households expense data is more satisfactory with the  Clayton family than for the other two copulas.\\

We now apply the maximum likelihood method for fitting copulas to comfort our above graphical results. For this, it suffices to compute (see Table \ref{tab2}), for each of the three copulas,  the log-likelihood function 
defined as
$$
L_C\left(\theta, \textbf{u}, \textbf{v}\right)=\sum_{i=1}^n \log C_{12}\left(u_i,v_i\right),
$$
where $C_{12}(u,v)=\frac{\partial}{\partial u}\frac{\partial}{\partial v} C(u,v)$, $\textbf{u}=(u_1,\cdots,u_n)$ and $\textbf{v}=(v_1,\cdots,v_n)$ .\\

From Table \ref{tab2} below,  we can conclude that the Clayton copula fits best our data. So we recommend it to model dependence of the senegalese households expense data in the years 2001 and 2006.\\

\begin{table}[htbp]
\centering
\begin{tabular}{|c|c|c|}
\hline Copula & Estimation of $\theta$ &  log-likelihood  \\
\hline Clayton & 1.38 & 547.61  \\
\hline Gumbel & 1.69 & 28.57  \\
\hline Frank & -0.57 & 299.78  \\
\hline 
\end{tabular}
\caption{Log-likelihood values }
\label{tab2}
\end{table}

\subsection{Concluding remarks}
This paper  presented a nonparametric method to estimate the copula function by providing asymptotic confidence bands based on the local linear kernel estimator. The results are applied to select graphically the best copula function that fits the dependence structure between pseudo-panel data. 

In perspective, similar results can be obtained with other kernel-type estimators of copula function like the mirror-reflection and transformation estimators.

\section{Proofs}
In this section, we first expose technical details allowing us to use the methodology of Mason (\cite{r3}) described in Proposition \ref{p1} and Corollary \ref{crl1} that are necessary to prove our results. In the second step we give successively the proofs of the theorems stated in Section 2.\\

We begin by decomposing the difference $\hat{C}_{n,h}^{(LL)}(u,v)- C(u,v)$, $(u,v)\in [0,1]^2$ as follows : 
$$
 \hat{C}_{n,h}^{(LL)}(u,v)-C(u,v)=\underbrace{\left[\hat{C}_{n,h}^{(LL)}(u,v)-\mathbb{E}\hat{C}_{n,h}^{(LL)}(u,v)\right]}_{Deviation}+\underbrace{\left[\mathbb{E}\hat{C}_{n,h}^{(LL)}(u,v)-C(u,v)\right]}_{Bias}
.$$
The probabilistic term $$\hat{D}_{n,h}(u,v):=\hat{C}_{n,h}^{(LL)}(u,v)-\mathbb{E}\hat{C}_{n,h}^{(LL)}(u,v)$$  is called the deviation of the estimator from its expectation. We'll study its behavior by making use of the methodology described in \cite{r3}. The other term that we denote 
$$ B_{n,h}(u,v):=\mathbb{E}\hat{C}_{n,h}^{(LL)}(u,v)-C(u,v)$$ is the so-called bias of the estimator. It is deterministic and its behavior will depend upon the smoothness conditions on the copula $C$ and the bandwidth $h$.\\
Recall the estimator proposed by Deheuvels in \cite{r7}, which is defined as 
$$ \hat{C}_n(u,v)=\frac{1}{n}\sum_{i=1}^n\mathbb{I}\{\hat{U}_i\leq u,\hat{V}_i\leq v\},\quad\hat{U}_i=F_n(X_i),\quad \hat{V}_i=G_n(Y_i), $$
where $F_n$ and $G_n$ are the empirical cumulative distribution functions of the marginals $F$ and $G$. This estimator is asymptotically equivalent (up to a term $O(n^{-1})$) with the estimator based directly on Sklar's Theorem given by
$$ C_n(u,v)= H_n(F_n^{-1}(u),G_n^{-1}(v)),$$
with $H_n$ the empirical joint distribution function of $(X,Y)$. Then the empirical copula process is defined as
$$\mathbb{C}_n(u,v)=\sqrt{n}[C_n(u,v)- C(u,v)].$$
To study the behavior of the deviation $\hat{D}_{n,h}(u,v)$, we introduce the following notation.
Let
$$ \tilde{C}_n(u,v)=\frac{1}{n}\sum_{i=1}^n\mathbb{I}\{U_i\leq u,V_i\leq v\}$$
be the uniform bivariate empirical distribution function based on a sample \\$(U_1,V_1),\cdots,(U_n,V_n)$ of i.i.d random  variables uniformly distributed on $[0,1]^2$. 
Define the following empirical process
 $$\mathbb{\tilde{C}}_n(u,v)=\sqrt{n}[\tilde{C}_n(u,v)- C(u,v)],\quad (u,v)\in[0,1]^2.$$
Then one can  observe that 
\begin{equation}\label{c3}
\mathbb{\tilde{C}}_n(u,v)=\mathbb{C}_n(u,v)+ \frac{1}{\sqrt{n}}.
\end{equation}
 For all $(u,v)\in [0,1]^2$, define 
\begin{eqnarray*}
g_{n,h}&=&\hat{C}_{n,h}^{(LL)}(u,v)-\tilde{C}_n(u,v)\\
&  =& \frac{1}{n}\sum_{i=1}^{n}\left [K_{u,h}\left(\frac{u-\hat{U}_i}{h} \right)K_{v,h}\left(\frac{v-\hat{V}_i}{h} \right)- \mathbb{I}\{U_i\leq u,V_i\leq v\}\right ]\\
& =&\frac{1}{n}\sum_{i=1}^{n}\left [K_{u,h}\left(\frac{u-\hat{F}_n oF^{-1}(U_i)}{h} \right)K_{v,h}\left(\frac{v-\hat{G}_n oG^{-1}(V_i)}{h} \right)-\mathbb{I}\{U_i\leq u,V_i\leq v\}\right ] \\
& =:& \frac{1}{n}\sum_{i=1}^{n}g(U_i,V_i,h),
\end{eqnarray*}
where $g$ belongs to a class of measurable functions $\mathcal{G}$ defined as 
$$
\mathcal{G}=\left\{\begin{array}{c} 
(s,t)\mapsto g(s,t,h)=K_{u,h}\left(\frac{u-\zeta_1(s)}{h} \right)K_{v,h}\left(\frac{v-\zeta_2(t)}{h}\right)- \mathbb{I}\{s\leq u,t\leq v\},\\
 u,v\in[0,1]^2, 0<h<1 \,\text{and }\, \zeta_1,\zeta_2:[0,1]\mapsto [0,1] \,\text{nondecreasing.}
\end{array} \right\}
$$
Since ${\tilde{C}}_n(u,v)$ is an unbiased estimator for $C(u,v)$, one can observe that 

$$
\sqrt{n}\vert g_{n,h}- \mathbb{E}g_{n,h}\vert= \vert\sqrt{n}\hat{D}_{n,h}(u,v)-\mathbb{\tilde{C}}_n(u,v)\vert.
$$

To make use of Mason's Theorem in \cite{r3}, the class of functions $\mathcal{G}$ must verify the following  four conditions : 
\begin{itemize}
\item[(G.i)]\ \ \ $\displaystyle \sup_{0\leq h\leq 1}\sup_{g\in \mathcal{G}}\left\|g\left(U,V,h\right)\right\|_\infty=:\kappa <\infty$
\item[(G.ii)]\ \ \ There exists some constant $C'>0$ such that for all $h\in [0,1]$,
\\$\displaystyle \sup_{g\in \mathcal{G}}\mathbb{E}\left[g^2\left(U,V,h\right)\right]\leq C'h $
\item[(F.i)]\ \ \ 
$\mathcal{G}$ satisfies the uniform entropy condition, i.e., $\exists \ \ C_0>0, \nu_0>0\ :\ N\left(\epsilon,\mathcal{G}\right)\leq C_0\epsilon^{-\nu_0}$.
\item[(F.ii)]\ \ \ $\mathcal{G}$ is a pointwise measurable class, i.e there exists a countable sub-class $\mathcal{G}_0$ of $\mathcal{G}$ such that for all $g\in \mathcal{G}$, there exits $\left(g_m\right)_m\subset \mathcal{G}_0$ such that $g_m\longrightarrow g.$
\end{itemize}

The checking of these conditions constitutes the proof of the following proposition which will be done in Appendix A.

\begin{proposition}\label{p1}
Suppose that the copula function $C$ has bounded first order partial derivatives on $(0, 1)^2$. Assuming (G.i), (G.ii), (F.i) and (F.ii), we have for $c > 0,\ 0 < h_0 < 1,$ with probability one 
$$
\limsup_{n\rightarrow\infty}\sup_{\frac{c\log n}{n}\leq h \leq h_0}\sup_{(u,v)\in(0,1)^2}
\frac{|\sqrt{n}\hat{D}_{n,h}(u,v)-\tilde{\mathbb{C}}_n(u,v)|}{\sqrt{h(|\log h\vert\vee\log\log n)}}=A(c),
$$
where $A(c)$ is a positive constant.
\end{proposition}

\begin{corollary}\label{crl1}
 Under the assumptions of Proposition \ref{p1}, for any sequence of constants $0<b_n<1,$ satisfying $\ b_n\rightarrow 0,\ b_n\geq (\log n)^{-1}$, one has with probability one 
 $$
 \sup_{\frac{c\log n}{n}\leq h \leq b_n}\sup_{(u,v)\in(0,1)^2}
\frac{|\sqrt{n}\hat{D}_{n,h}(u,v)-\tilde{\mathbb{C}}_n(u,v)|}{\sqrt{\log\log n}}=O(\sqrt{b_n}).
 $$
 \end{corollary}
\begin{proof}{( \textbf{Corollary \ref{crl1})}}\\
First, observe that the condition $b_n\geq (\log n)^{-1}$ yields
\begin{equation}\label{dc}
\frac{\vert \log b_n\vert}{\log\log n}\leq 1.
\end{equation}
Next, by the monotonicity of the function $x\mapsto x\vert\log x \vert$ on $[0,1/e]$, one can write for $n$ large enough, $h\vert\log h \vert\leq b_n\vert\log b_n \vert$ and hence,
\begin{equation}
h(\vert\log h \vert  \vee\log\log n)\leq b_n(\vert\log b_n \vert \vee\log\log n).
\end{equation}
Combining this and Proposition \ref{p1}, we obtain 
$$
 \sup_{\frac{c\log n}{n}\leq h \leq b_n}\sup_{(u,v)\in(0,1)^2}
\frac{|\sqrt{n}\hat{D}_{n,h}(u,v)-\tilde{\mathbb{C}}_n(u,v)|}{\sqrt{b_n\log\log n\left(\frac{\vert \log b_n\vert}{\log\log n}\vee 1\right)}}=O(1).
 $$
Thus the Corrollary \ref{crl1} follows from  \eqref{dc}.

\end{proof}

\begin{proof}{(\textbf{Theorem \ref{t1}})}\\
The proof is based upon an approximation of the empirical copula process $\mathbb{C}_n$ by a Kiefer process (see \cite{r18}, p. 100). Let $\mathbb{W}(u,v,t)$ be a $3$-parameters  Wiener process defined on $[0,1]^2\times[0,\infty)$. Then the gaussian process $\mathbb{K}(u,v,t)=\mathbb{W}(u,v,t)-\mathbb{W}(1,1,t).uv$ is called a $3$-parameters  Kiefer process defined on $[0,1]^2\times[0,\infty)$.\\

By Theorem 3.2 in \cite{r18}, for $d=2$, there exists a sequence of gaussian processes $\left\{\mathbb{K}_{C}(u,v,n), u,v\in[0,1], n>0\right\}$ such that
$$ \sup_{(u,v)\in[0,1]^2}\left|\sqrt{n}\mathbb{C}_n(u,v)-\mathbb{K}_{C}^\ast(u,v,n)\right|=O\left(n^{3/8}(\log n)^{3/2}\right),$$ where 
$$\mathbb{K}_{C}^\ast(u,v,n)=\mathbb{K}_{C}(u,v,n)-\mathbb{K}_{C}(u,1,n)\frac{\partial C(u,v)}{\partial u}-\mathbb{K}_\mathbb{C}(1,v,n)\frac{\partial C(u,v)}{\partial v}.$$
This yields
\begin{equation}\label{c1}
\limsup_{n\rightarrow\infty}\sup_{(u,v)\in[0,1]^2}\frac{\left|\mathbb{C}_n(u,v)\right|}{\sqrt{2\log\log n}}=\limsup_{n\rightarrow\infty}\sup_{(u,v)\in[0,1]^2}\frac{\left|\mathbb{K}_{C}^\ast(u,v,n)\right|}{\sqrt{2n\log\log n}}.
\end{equation}
By the works of Wichura on the  iterated law of logarithm (see \cite{r19}),
one has
\begin{equation}\label{c2}
\limsup_{n\rightarrow\infty}\sup_{(u,v)\in[0,1]^2}\frac{\left|\mathbb{K}_\mathbb{C}^\ast(u,v,n)\right|}{\sqrt{2n\log\log n}}\leq 3,
\end{equation}

\noindent which readily implies  
$$\limsup_{n\rightarrow\infty}\sup_{(u,v)\in[0,1]^2}\frac{\left|\mathbb{C}_n(u,v)\right|}{\sqrt{2\log\log n}}\leq 3.$$
Since ${\mathbb{C}}_n(u,v)$ and $\tilde{\mathbb{C}}_n(u,v)$ are asymptotically equivalent in view of (\ref{c3}), one obtains
$$\limsup_{n\rightarrow\infty}\sup_{(u,v)\in[0,1]^2}\frac{\left|\tilde{\mathbb{C}}_n(u,v)\right|}{\sqrt{2\log\log n}}\leq 3.$$
The proof is then finished by applying Corollary \ref{crl1} which yields
\begin{equation}\label{lil}
\limsup_{n\rightarrow\infty}\sup_{\frac{c\log n}{n}\leq h \leq b_n}\sup_{(u,v)\in[0,1]^2}\frac{\left|\sqrt{n}\hat{D}_{n,h}(u,v)\right|}{\sqrt{2\log\log n}}\leq 3.
\end{equation}
Thus, there exists a contant $A(c)$, with $0<A(c)\leq 3$, such that 
\begin{equation}\label{lilo}
\limsup_{n\rightarrow\infty}\left\{\left(\frac{n}{2\log\log n}\right)^{1/2}\sup_{\frac{c\log n}{n}\leq h \leq b_n}\sup_{(u,v)\in[0,1]^2}\left|\hat{C}_{n,h}^{(LL)}(u,v)-\mathbb{E}\hat{C}_{n,h}^{(LL)}(u,v)\right|\right\}=A(c).
\end{equation}
\end{proof}
 
\begin{proof} (\textbf{ Theorem \ref{t2}})\\
 For all $(u,v)\in [0,1]^2$, one has
$$
C(u,v)=\int_{-1}^{1}\int_{-1}^{1}C(u,v)k_{u,h}(s)k_{v,h}(t)dsdt
$$
and
$$
\mathbb{E}\hat{C}_{n,h}^{(LL)}(u,v)= \mathbb{E}\left[K_{u,h}\left(\frac{u-\zeta_{1,n}(U_i)}{h} \right)K_{v,h}\left(\frac{v-\zeta_{2,n}(V_i)}{h}\right)\right],
$$
with $\zeta_{1,n}(U_i)=\hat{F}_n oF^{-1}(U_i)=\hat{U}_i$  and  $\,\,\zeta_{2,n}(V_i)=\hat{G}_n oG^{-1}(V_i)=\hat{V}_i$.
We can easily show that 
$$
\mathbb{E}\hat{C}_{n,h}^{(LL)}(u,v)=\int_{-1}^{1}\int_{-1}^{1}C\left(\zeta_{1,n}^{-1}(u-sh),\zeta_{2,n}^{-1}(v-th)\right)k_{u,h}(s)k_{v,h}(t)dsdt.
$$ 
Hence
\begin{eqnarray*}
B_{n,h}(u,v)& = & \mathbb{E}\hat{C}_{n,h}^{(LL)}(u,v) - C(u,v)\\
        & = &\int_{-1}^{1}\int_{-1}^{1}\left[ C\left(\zeta_{1,n}^{-1}(u-sh),\zeta_{2,n}^{-1}(v-th) - C(u,v)\right)\right] k_{u,h}(s)k_{v,h}(t)dsdt.
\end{eqnarray*}
By continuity of $F$ and $G$, we have for $n$ large enough, 
$$
\zeta_{1,n}^{-1}(u-sh)= Fo\hat{F}_n^{-1}(u-sh)\sim u-sh
$$ and
$$
\zeta_{2,n}^{-1}(v-th)= Go\hat{G}_n^{-1}(v-th)\sim v-th.
$$      
Thus,
$$
B_{n,h}(u,v)= \int_{-1}^{1}\int_{-1}^{1}\left[ C(u-sh,v-th) - C(u,v)\right] k_{u,h}(s)k_{v,h}(t)dsdt.
$$
By applying a 2-order Taylor expansion  and taking account of the symmetry of the kernels $k_{u,h}(.)$ and $k_{v,h}(.)$   i.e, 
$$ \int_{-1}^{1}sk_{u,h}(s)ds=0\quad \text{and} \quad \int_{-1}^{1}tk_{v,h}(t)dt=0,$$
we obtain, by Fubini, that for all $(u,v)\in [0,1]^2$,
 \begin{eqnarray*}
 B_{n,h}(u,v)&=& h^2\int_{-1}^{1}\int_{-1}^{1}\left[ s^2\frac{\partial^2 C(u,v)}{\partial u^2} + st\frac{\partial^2 C(u,v)}{\partial u\partial v}+ s^2\frac{\partial^2 C(u,v)}{\partial v^2}\right] k_{u,h}(s)k_{v,h}(t)dsdt\\
 & = & h^2\left[\int_{-1}^{1} s^2 \frac{\partial^2 C(u,v)}{\partial u^2}k_{u,h}(s)ds + \int_{-1}^{1}t^2\frac{\partial^2 C(u,v)}{\partial v^2} k_{v,h}(t)dt\right].
 \end{eqnarray*}
Since the second order partial derivatives are assumed to be bounded, then we can infer that
$$
\sup_{\frac{c\log n}{n}\leq h \leq b_n}\sup_{(u,v)\in[0,1]^2}B_{n,h}(u,v)= O(b_n^2)
$$
and hence,
\begin{equation} 
\left(\frac{n}{2\log\log n}\right)^{1/2}\sup_{\frac{c\log n}{n}\leq h \leq b_n}\sup_{(u,v)\in[0,1]^2}B_{n,h}(u,v)=O\left(\frac{\sqrt{n}b_n^2}{\sqrt{2\log\log n}}\right)=o(1). 
\end{equation}
\end{proof}

\begin{proof} (\textbf{ Theorem \ref{t3}})\\
It suffices to show \eqref{leq1} and \eqref{leq2}. From \eqref{e6}, we can infer that for any given $\epsilon>0$ and $\delta>0$, there exists $N\in\mathbb{N}$ such that for all $n>N$,
$$
\mathbb{P}\left(\frac{R_n}{A(c)}\left|\mathbb{E}\hat{C}_{n,\hat{h}_n(u,v)}^{(LL)}(u,v)-C(u,v)\right|>\epsilon,\;\forall (u,v)\in[0,1]^2\right)>\delta.
$$ 
That is
\begin{equation}\label{esp1}
\mathbb{P}\left(\mathbb{E}\hat{C}_{n,\hat{h}_n(u,v)}^{(LL)}(u,v)-\epsilon\frac{A(c)}{R_n}\leq C(u,v)\leq \mathbb{E}\hat{C}_{n,\hat{h}_n(u,v)}^{(LL)}(u,v)+\epsilon\frac{A(c)}{R_n},\;\forall (u,v)\in[0,1]^2\right)>1-\delta.
\end{equation}
On the other hand we deduce from \eqref{e5} that for all $(u,v)\in[0,1]^2$,
\begin{equation*}
\mathbb{E}\hat{C}_{n,\hat{h}_n(u,v)}^{(LL)}(u,v)=\hat{C}_{n,\hat{h}_n(u,v)}^{(LL)}(u,v)\pm \left(1+o_p(1)\right)\frac{A(c)}{R_n}.
\end{equation*}
\begin{case}
If 
\begin{equation*}
\mathbb{E}\hat{C}_{n,\hat{h}_n(u,v)}^{(LL)}(u,v)=\hat{C}_{n,\hat{h}_n(u,v)}^{(LL)}(u,v)-\left(1+o_p(1)\right)\frac{A(c)}{R_n},
\end{equation*}
\end{case}
 then \eqref{esp1} becomes
\begin{eqnarray*}
\mathbb{P}\bigg(\hat{C}_{n,\hat{h}_n(u,v)}^{(LL)}(u,v)&-&E_n(\epsilon)-o_p(1)\frac{A(c)}{R_n}\leq C(u,v) \\
&\leq& \hat{C}_{n,\hat{h}_n(u,v)}^{(LL)}(u,v)-\Delta_n(\epsilon)-o_p(1)\frac{A(c)}{R_n},\;\forall (u,v)\in[0,1]^2\bigg)>1-\delta.
\end{eqnarray*}
Thus, for any given $\tau \in (0,1)$ and all large $n$, we can write
\begin{eqnarray}\label{esp2}
\mathbb{P}\bigg(\hat{C}_{n,\hat{h}_n(u,v)}^{(LL)}(u,v)&-&E_n(\epsilon)-\tau\frac{A(c)}{R_n}\leq C(u,v) \nonumber \\
&\leq& \hat{C}_{n,\hat{h}_n(u,v)}^{(LL)}(u,v)-\Delta_n(\epsilon)-\tau\frac{A(c)}{R_n},\;\forall (u,v)\in[0,1]^2\bigg)>1-\delta.
\end{eqnarray}
\begin{case}
If $$\mathbb{E}\hat{C}_{n,\hat{h}_n(u,v)}^{(LL)}(u,v)=\hat{C}_{n,\hat{h}_n(u,v)}^{(LL)}(u,v)+\left(1+o_p(1)\right)\frac{A(c)}{R_n},$$
\end{case}
then, analogously to \textbf{Case 1}, we  can infer from \eqref{esp1} that, for any given $\tau \in (0,1)$ and all large $n$,
\begin{eqnarray}\label{esp3}
\mathbb{P}\bigg(\hat{C}_{n,\hat{h}_n(u,v)}^{(LL)}(u,v)&+&\Delta_n(\epsilon)+\tau\frac{A(c)}{R_n}\leq C(u,v) \nonumber\\
&\leq& \hat{C}_{n,\hat{h}_n(u,v)}^{(LL)}(u,v)+E_n(\epsilon)+\tau\frac{A(c)}{R_n},\;\forall (u,v)\in[0,1]^2\bigg)>1-\delta.
\end{eqnarray}
Letting  $\tau$ tends to $0$, it follows from \eqref{esp2} and \eqref{esp3} that

\begin{equation}\label{esp22}
\mathbb{P}\left(\hat{C}_{n,\hat{h}_n(u,v)}^{(LL)}(u,v)-E_n(\epsilon)\leq C(u,v)\leq \hat{C}_{n,\hat{h}_n(u,v)}^{(LL)}(u,v)-\Delta_n(\epsilon),\;\forall (u,v)\in[0,1]^2\right)>1-\delta
\end{equation}
or
\begin{equation}\label{esp33}
\mathbb{P}\left(\hat{C}_{n,\hat{h}_n(u,v)}^{(LL)}(u,v)+\Delta_n(\epsilon)\leq C(u,v)\leq \hat{C}_{n,\hat{h}_n(u,v)}^{(LL)}(u,v)+E_n(\epsilon),\;\forall (u,v)\in[0,1]^2\right)>1-\delta.
\end{equation}
Now, by observing that $$\hat{C}_{n,\hat{h}_n(u,v)}^{(LL)}(u,v)-\Delta_n(\epsilon)\leq \hat{C}_{n,\hat{h}_n(u,v)}^{(LL)}(u,v)+\Delta_n(\epsilon),\ \forall \epsilon,$$
we can write, for any  $\epsilon>0$, with  probability tending to 1,
$$C(u,v)\in\left[\hat{C}_{n,\hat{h}_n(u,v)}^{(LL)}(u,v)-E_n(\epsilon),\hat{C}_{n,\hat{h}_n(u,v)}^{(LL)}(u,v)+E_n(\epsilon)\right],\ \forall (u,v)\in [0,1]^2$$

and

$$C(u,v)\notin \left[\hat{C}_{n,\hat{h}_n(u,v)}^{(LL)}(u,v)-\Delta_n(\epsilon),\hat{C}_{n,\hat{h}_n(u,v)}^{(LL)}(u,v)+\Delta_n(\epsilon)\right],\ \forall (u,v)\in [0,1]^2.$$
That is, \eqref{leq1} and \eqref{leq2} hold.
\end{proof}




 

\appendix{\textbf{\Large{Appendix A} : Proof of Proposition 1}}

\begin{proof}
It suffices to check the conditions (G.i), (G.ii), (F.i) and (F.ii) given in section 3.\\
\noindent \textbf{Checking for (G.i).} For $(u,v)\in[0,1]^2$ and $g\in\mathcal{G}$, one has 
\begin{eqnarray*}
 g\left(U_i,V_i,h\right)& = & K_{u,h}\left(\frac{u-\hat{U}_i}{h}\right)K_{v,h}\left(\frac{v-\hat{V}_i}{h}\right)-\mathbb{I}\{U_i\leq u,V_i\leq v\}\\
 &= & \int_{-\infty}^{\frac{u-\hat{U}_i}{h}}k_{u,h}(t)dt \int_{-\infty}^{\frac{v-\hat{V}_i}{h}}k_{v,h}(s)ds-\mathbb{I}\{U_i\leq u,V_i\leq v\}\\
& = & \int_{-1}^{1}\int_{-1}^{1}\mathbb{I}_{\left\{\hat{U}_i\leq u-th ,\hat{V}_i\leq v-sh\right\}}\;k_{u,h}(t)\;k_{v,h}(s)dt\;ds-\mathbb{I}\{U_i\leq u,V_i\leq v\}\\
&\leq & \int_{-1}^{1}\int_{-1}^{1}k_{u,h}(t)\;k_{v,h}(s)dt\;ds-\mathbb{I}\{U_i\leq u,V_i\leq v\}\ \leq 4\ \left\|k_{u,h}\right\|\left\|k_{v,h}\right\|+1.
\end{eqnarray*}

Then, $$ \left\|g\left(U_i,V_i,h\right)\right\|_\infty \leq 4\left\|k_{u,h}\right\|\left\|k_{v,h}\right\|+1=:\kappa. $$ 
This implies $$ \sup_{0\leq h\leq 1}\sup_{g\in \mathcal{G}}\left\|g\left(U_i,V_i,h\right)\right\|_\infty=:\kappa <\infty. $$

\noindent \textbf{Checking for (G.ii).}\\
 We have to show that
 $$ \sup_{g\in \mathcal{G}}\mathbb{E}g^2(U,V,h)\leq C'h ,$$ where $C'$ is a constant. Recall that $\zeta_{1,n}^{-1}(u)= Fo\hat{F}_n^{-1}(u)$ and $\zeta_{2,n}^{-1}(v-th)= Go\hat{G}_n^{-1}(v).$ Then, we can write
\begin{eqnarray*}
\mathbb{E}g^2(U,V,h) &= & \mathbb{E}\left[K_{u,h}\left(\frac{u-\zeta_{1,n}(U)}{h} \right)K_{v,h}\left(\frac{v-\zeta_{2,n}(V)}{h}\right)- \mathbb{I}\{U\leq u,V\leq v\}\right]^2\\
&= & \mathbb{E}\left[K_{u,h}\left(\frac{u-\zeta_{1,n}(U)}{h} \right)K_{v,h}\left(\frac{v-\zeta_{2,n}(V)}{h}\right)\right]^2\\
& & - 2\mathbb{E}\left[K_{u,h}\left(\frac{u-\zeta_{1,n}(U)}{h} \right)K_{v,h}\left(\frac{v-\zeta_{2,n}(V)}{h}\right)\mathbb{I}\{U\leq u,V\leq v\}\right] +C(u,v)\\
& =: & A -2B + C(u,v).
\end{eqnarray*}
Now we express $A$ and $B$ as  integrals of the copula function $C(u,v)$.
\begin{eqnarray*}
B & = &\mathbb{E}\left[K_{u,h}\left(\frac{u-\zeta_{1,n}(U)}{h} \right)K_{v,h}\left(\frac{v-\zeta_{2,n}(V)}{h}\right)\mathbb{I}\{U\leq u,V\leq v\}\right]\\
&=& \mathbb{E}\left[\int_{-\infty}^\frac{u-\zeta_{1,n}(U)}{h}\int_{-\infty}^\frac{v-\zeta_{2,n}(V)}{h}k_{u,h}(s)k_{v,h}(t)\mathbb{I}\{U\leq u,V\leq v\}dsdt\right]\\
&=& \mathbb{E}\left[ \int_{-1}^{1}\int_{-1}^{1}k_{u,h}(s)k_{v,h}(t)\mathbb{I}\{s\leq \frac{u-\zeta_{1,n}(U)}{h},t\leq \frac{u-\zeta_{2,n}(V)}{h}\} \mathbb{I}\{U\leq u,V\leq v\}dsdt\right]\\
&=& \mathbb{E}\left[ \int_{-1}^{1}\int_{-1}^{1}k_{u,h}(s)k_{v,h}(t)\mathbb{I}\{U\leq \min(u,\zeta_{1,n}^{-1}(u-sh)),V\leq \min(v,\zeta_{2,n}^{-1}(v-th))\}dsdt \right]\\
&=& \int_{-1}^{1}\int_{-1}^{1}k_{u,h}(s)k_{v,h}(t)C[\min(u,\zeta_{1,n}^{-1}(u-sh)),\min(v,\zeta_{2,n}^{-1}(v-th))]dsdt.
\end{eqnarray*}

\noindent Since because $K_{u,h}(\cdot)$  takes its values in [0,1] as a distribution function, we observing that $K_{u,h}^2(x)\leq K_{u,h}(x)$. Then we can write
\begin{eqnarray*}
A & = &  \mathbb{E}\left[K_{u,h}^2\left(\frac{u-\zeta_{1,n}(U)}{h} \right)K_{v,h}^2\left(\frac{v-\zeta_{2,n}(V)}{h}\right)\right]\\
&\leq & \mathbb{E}\left[K_{u,h}\left(\frac{u-\zeta_{1,n}(U)}{h} \right)K_{v,h}\left(\frac{v-\zeta_{2,n}(V)}{h}\right)\right]\\
& \leq &\mathbb{E}\left[ \int_{-1}^{1}\int_{-1}^{1}k_{u,h}(s)k_{v,h}(t)\mathbb{I}\{U\leq \zeta_{1,n}^{-1}(u-sh),V\leq \zeta_{2,n}^{-1}(v-th)\}dsdt \right]\\
&\leq & \int_{-1}^{1}\int_{-1}^{1}k_{u,h}(s)k_{v,h}(t)C[\zeta_{1,n}^{-1}(u-sh),\zeta_{2,n}^{-1}(v-th)]dsdt.
\end{eqnarray*}
We can also notice that $$ C(u,v)=\int_{-1}^{1}\int_{-1}^{1}k_{u,h}(s)k_{v,h}(t)C(u,v)dsdt.$$ Thus
\begin{eqnarray*}
\mathbb{E}g^2(U,V,h)&\leq &\int_{-1}^{1}\int_{-1}^{1}k_{u,h}(s)k_{v,h}(t)C[\zeta_{1,n}^{-1}(u-sh),\zeta_{2,n}^{-1}(v-th)]dsdt\\
 & & -2\int_{-1}^{1}\int_{-1}^{1}k_{u,h}(s)k_{v,h}(t)C[\min(u,\zeta_{1,n}^{-1}(u-sh)),\min(v,\zeta_{2,n}^{-1}(v-th))]dsdt\\
 &  & +\int_{-1}^{1}\int_{-1}^{1}k_{u,h}(s)k_{v,h}(t)C(u,v)dsdt
\end{eqnarray*}
For $n$ enough large, we have by continuity of $F$ and $G$,
$$
\zeta_{1,n}^{-1}(u-sh)= Fo\hat{F}_n^{-1}(u-sh)\sim u-sh
$$ and
$$
\zeta_{2,n}^{-1}(v-th)= Go\hat{G}_n^{-1}(v-th)\sim v-th.
$$
By splitting the integrals, we obtain after simple calculus that
\begin{eqnarray*}
\mathbb{E}g^2(U,V,h)&\leq &\int_{-1}^{0}\int_{-1}^{0}k_{u,h}(s)k_{v,h}(t)[C(u-sh,v-th)-C(u,v)]dsdt\\
& & + \int_{0}^{1}\int_{0}^{1}k_{u,h}(s)k_{v,h}(t)[C(u-sh,v-th)-C(u,v)]dsdt\\
 & & \int_{-1}^{0}\int_{0}^{1}k_{u,h}(s)k_{v,h}(t)[C(u-sh,v-th)-C(u,v-th)]dsdt \\
 & &+\int_{-1}^{0}\int_{0}^{1}k_{u,h}(s)k_{v,h}(t)[C(u,v)-C(u,v-th)]dsdt\\
 & & +\int_{0}^{1}\int_{-1}^{0}k_{u,h}(s)k_{v,h}(t)[C(u-sh,v-th)-C(u-sh,v)]dsdt\\
& & +\int_{0}^{1}\int_{-1}^{0}k_{u,h}(s)k_{v,h}(t)[C(u,v)-C(u-sh,v)]dsdt\\
&=: & T_1 + T_2+T_3+T_4+ T_5+ T_6.\\ 
\end{eqnarray*}
 All these six terms can be bounded up by applying Taylor expansion. Precisely, we have
 \begin{eqnarray*}
 |T_1|& \leq & h(\Vert C_u(u,v)\Vert +\Vert C_v(u,v)\Vert) \\
 |T_2|&\leq & h(\Vert C_u(u,v)\Vert +\Vert C_v(u,v)\Vert)\\
 |T_3|& \leq & h(\Vert C_u(u,v)\Vert)\\
  |T_4| &\leq & h(\Vert C_v(u,v)\Vert)\\
  |T_5|& \leq & h(\Vert C_v(u,v)\Vert)\\
  |T_6|& \leq  & h(\Vert C_u(u,v)\Vert).\\
 \end{eqnarray*}
 From this, we can conclude that 
$$
\mathbb{E}g^2(U,V,h)\leq h 4(\Vert C_u(u,v)\Vert +\Vert C_v(u,v)\Vert),
$$ 
and $$\sup_{g\in \mathcal{G}}\mathbb{E}\left[g^2\left(U,V,h\right)\right]\leq C'h,$$ 
with $C'=4(\Vert C_u(u,v)\Vert +\Vert C_v(u,v)\Vert).$ \\

\textbf{Checking for (F.i)}. We have to check that $\mathcal{G}$ satisfies the uniform entropy condition.\\
Consider the following classes of functions :\\
\noindent
$ \mathbb{F}=\left\{\lambda x+m\ ;\lambda\geq 1,\;m\in\mathbb{R}\right\}$\\
$\displaystyle \mathbb{K}_0=\left\{K(\lambda x+m)\ ;\lambda\geq1,\;m\in\mathbb{R}\right\}$\\
$\displaystyle \mathbb{K}=\left\{K(\lambda x+m)K(\lambda y+m)\ ;\lambda\geq 1,\;m\in\mathbb{R}\right\}$\\
$\displaystyle \mathbb{H}=\left\{K(\lambda x+m)K(\lambda y+m)-\mathbb{I}\left\{x\leq u,y\leq v\right\}\ ;\lambda\geq 1 \;m\in\mathbb{R}, (u,v)\in [0,1]^2\right\}$.\\

It is clear that by applying the lemmas 2.6.15 and 2.6.18 in van der Vaart and Wellner (see \cite{r20}, p. 146-147), the sets $\mathbb{F},\;\mathbb{K}_0,\;\mathbb{K},\;\mathbb{H}$ are all VC-subgraph classes. Thus, by taking  the function $(x,y)\mapsto G(x,y)=\left\|k\right\|^2+1 $ as a measurable envelope function for $\mathbb{H}$ ( indeed $G(x,y)\geq \sup_{g\in\mathbb{H}}\left|g(x,y)\right|,\ \forall (x,y))$, we can infer from Theorem 2.6.7 in \cite{r20} that  $\mathbb{H}$ satisfies the uniform entropy condition. Since $\mathbb{H}$ and $\mathcal{G}$ have the same structure, we can conclude that $\mathcal{G}$ satisfies this property too. That is,

$$\exists \ \ C_0>0, \nu_0>0\ :\ N\left(\epsilon,\mathcal{G}\right)\leq C_0\epsilon^{-\nu_0},\quad 0<\epsilon<1.$$

\noindent \textbf{Checking for (F.ii).}\\

\noindent Define the class of functions 
$$
\mathcal{G}_0=\left\{K_{u,h}\left(\frac{u-\zeta_1(s)}{h} \right)K_{v,h}\left(\frac{v-\zeta_2(t)}{h}\right)- \mathbb{I}\{s\leq u,t\leq v\}; u,v\in\left([0,1]\cap\mathbb{Q}\right)^2; 0<h<1  \right\}.
$$
It's clear that  $\mathcal{G}_0$ is countable and  $\mathcal{G}_0\subset\mathcal{G}$. Let
$$g(x,y)=K_{u,h}\left(\frac{u-\zeta_1(x)}{h} \right)K_{v,h}\left(\frac{v-\zeta_2(y)}{h}\right)- \mathbb{I}\{x\leq u,y\leq v\}\in\mathcal{G}, (x,y)\in [0,1]^2$$
and, for $m\geq 1$, 
$$g_m(x,y)=K_{u_m,h}\left(\frac{u_m-\zeta_1(x)}{h} \right)K_{v_m,h}\left(\frac{v_m-\zeta_2(y)}{h}\right)- \mathbb{I}\{x\leq u_m,y\leq v_m\},$$
where $ u_m=\frac{1}{m^2}[m^2u]+\frac{1}{m^2}$ and $ v_m=\frac{1}{m^2}[m^2v]+\frac{1}{m^2}$.\\

Let $\alpha_m=u_m-u,\quad \beta_m=v_m-v$ and define
$$\delta_{m,u}=\left(\frac{u_m-\zeta_1(x)}{h} \right)-\left(\frac{u-\zeta_1(x)}{h} \right)=\frac{u_m-u}{h}=\frac{\alpha_m}{h} $$ and
$$\delta_{m,v}=\left(\frac{v_m-\zeta_2(y)}{h} \right)-\left(\frac{v-\zeta_2(y)}{h} \right)=\frac{v_m-v}{h}=\frac{\beta_m}{h}.$$
Then, one can easily see that  $\displaystyle 0<\alpha_m\leq\frac{1}{m^2}$ and $\displaystyle 0<\beta_m\leq\frac{1}{m^2}$.\\
This implies, for all large  $m$, that\
$\delta_{m,u}\searrow 0$ and $\delta_{m,v}\searrow 0$, 
which are equivalent to \\$ \left(\frac{u_m-\zeta_1(x)}{h} \right)\searrow \left(\frac{u-\zeta_1(x)}{h} \right)$ and $\left(\frac{v_m-\zeta_2(y)}{h} \right)\searrow \left(\frac{v-\zeta_2(y)}{h} \right).$\\
By right-continuity of $K_{w,h}$, we obtain
$$\forall (x,y)\in[0,1]^2, g_m(x,y)\longrightarrow g(x,y), m\rightarrow \infty$$
and conclude that $\mathcal{G}$ is pointwise measurable class.

\end{proof}

\appendix{\textbf{\Large{Appendix B}: Numerical computations}}

\begin{table}[ht]
\centering
\begin{tabular}{|c|c|c|c|c|}
\hline 
  & $(u,v)$ & $\alpha_n(u,v)$ & $C(u,v)$ &  $\beta_n(u,v)$ \\
\hline
 & (0.58,0.12) & -0.320 & 0.097 & 0.701  \\
 &(0.54,0.47) & -0.139 & 0.302 & 0.887  \\
 &(0.07,0.50) & -0.455 & 0.056 & 0.857  \\
 &(0.21,0.59) & -0.355 & 0.162 & 0.672  \\
 $\theta$ =0.5 &(0.18,0.42) & -0.441 & 0.118 & 0.585  \\
 &(0.42,0.52) & -0.204 & 0.268 & 1.032  \\
 &(0.73,0.61) &  0.005 & 0.475 & 0.887  \\
 &(0.07,0.96) & -0.271 & 0.069 & 0.755   \\
 &(0.69,0.85) &  0.090 & 0.602 & 1.118   \\
 &(0.28,0.72) & -0.298 & 0.233 & 0.729  \\
\hline 
 &(0.96,0.29) &  0.081 & 0.289 & 1.100  \\
 &(0.16,0.45) & -0.362 & 0.152 & 0.664 \\
 &(0.67,0.75) &  0.066 & 0.576 & 1.093  \\
 &(0.22,0.47) & -0.320 & 0.203 & 0.706  \\
 $\theta=2$ &(0.12,0.55) & -0.240 & 0.118 & 0.786  \\
 &(0.85,0.80) &  0.082 & 0.716 & 1.109  \\
 &(0.46,0.42) & -0.174 & 0.326 & 0.852  \\
 &(0.83,0.22) & -0.327 & 0.217 & 0.700   \\
 &(0.65,0.26) & -0.223 & 0.248 & 0.804   \\
 &(0.11,0.30) & -0.251 & 0.103 & 0.776  \\
\hline
 &(0.32,0.40) & -0.163 & 0.307 & 0.863  \\
 &(0.85,0.65) &  0.073 & 0.638 & 1.100  \\
 &(0.31,0.70) & -0.143 & 0.309 & 0.884  \\
 &(0.51,0.60) &  0.015 & 0.484 & 1.012  \\
 $\theta$=6 &(0.89,0.14) &  0.139 & 0.139 & 1.267  \\
 &(0.80,0.66) &  0.131 & 0.637 & 1.158  \\
 &(0.24,0.87) & -0.244 & 0.239 & 0.782  \\
 &(0.10,0.13) & -0.283 & 0.096 & 0.744   \\
 &(0.31,0.08) & -0.333 & 0.079 & 0.694   \\
 &(0.65,0.53) &  0.009 & 0.509 & 1.018  \\
\hline 
\end{tabular}
\caption{Confidence bands for Clayton copula calculated for some random couples of values $(u,v)$.}
\label{tab3}
\end{table}

\begin{table}[ht]
\centering
\begin{tabular}{|c|c|c|c|c|}
\hline
 & $(u,v)$ & $\alpha_n(u,v)$ & $C(u,v)$ &  $\beta_n(u,v)$ \\
\hline 
&(0.48,0.25) & -0.398 & 0.075 & 0.628  \\
& (0.12,0.80) & -0.338 & 0.077 & 0.689  \\
& (0.35,0.68) & -0.337 & 0.189 & 0.688  \\
& (0.73,0.21) & -0.317 & 0.119 & 0.710  \\
$\theta=-2$ & (0.74,0.44) &  0.203 & 0.279 & 0.823  \\
& (0.14,0.67) & -0.131 & 0.066 & 1.158  \\
& (0.56,0.81) &  0.070 & 0.418 & 0.950  \\
& (0.98,0.21) &  0.077 & 0.202 & 1.036   \\
& (0.29,0.25) & -0.461 & 0.038 & 0.566   \\
& (0.77,0.22) & -0.316 & 0.137 & 0.714 \\
\hline 
& (0.47,0.38) & -0.294 & 0.297 & 0.733  \\
& (0.13,0.86) & -0.398 & 0.128 & 0.628  \\
& (0.68,0.02) & -0.478 & 0.019 & 0.548  \\
$\theta=5$ & (0.34,0.20) & -0.373 & 0.146 & 0.654  \\
& (0.41,0.47) & -0.254 & 0.315 & 0.773  \\
& (0.53,0.24) & -0.354 & 0.212 & 0.673  \\
& (0.38,0.33) & -0.343 & 0.235 & 0.683  \\
& (0.31,0.60) &  0.233 & 0.280 & 0.793   \\
& (0.21,0.97) &  0.032 & 0.209 & 0.994   \\
& (0.07,0.48) & -0.356 & 0.063 & 0.670 \\
\hline 
& (0.87,0.45) & -0.234 & 0.449 & 0.793  \\
& (0.78,0.44) & -0.240 & 0.439 & 0.786  \\
& (0.60,0.72) &  0.044 & 0.593 & 0.983  \\
& (0.43,0.57) & -0.214 & 0.425 & 0.812  \\
$\theta=18$ & (0.26,0.33) & -0.298 & 0.246 & 0.729  \\
& (0.10,0.90) &  0.052 & 0.100 & 0.975  \\
& (0.86,0.40) & -0.285 & 0.399 & 0.741  \\
& (0.46,0.62) &  0.194 & 0.456 & 0.833   \\
& (0.05,0.52) & -0.449 & 0.049 & 0.578   \\
& (0.45,0.04) & -0.416 & 0.039 & 0.611 \\
\hline 
\end{tabular}
\caption{Confidence bands for Frank copula calculated for some random couples of values $(u,v)$ .}
\label{tab4}
\end{table}

\end{document}